\documentclass[a4paper, 11pt]{article}

\usepackage[round]{natbib}
\usepackage[utf8]{inputenc}
\usepackage{subfigure}
\usepackage{tikz}
\usetikzlibrary{shapes,arrows,positioning}
\usepackage{graphicx} 
\usepackage{nicefrac} 
\usepackage{enumitem}
\usepackage{amsmath}
\usepackage{amsthm}
\usepackage[T1]{fontenc}
\usepackage{amsfonts}
\usepackage[toc,page]{appendix}
\usepackage{authblk}

\newcommand{\transp}{\mathsf{T}}
\DeclareMathOperator{\Tr}{Tr}
\newcommand{\process}[1] { \{ #1 \}_{t\geq 1}}
\theoremstyle{plain}
\newtheorem{theorem}{Theorem}[section]
\newtheorem{corollary}[theorem]{Corollary}
\newtheorem{lemma}[theorem]{Lemma}

\newtheorem{assumption}{Assumption}
\theoremstyle{definition}
\newtheorem{definition}{Definition}
\theoremstyle{remark}

\tikzstyle{block} = [draw, fill=blue!20, rectangle, minimum height=3em, minimum width=6em]
\tikzstyle{sum} = [draw, fill=blue!20, circle, node distance=1cm]
\tikzstyle{input} = [coordinate]
\tikzstyle{output} = [coordinate]
\tikzstyle{pinstyle} = [pin edge={to-,thin,black}]
\tikzstyle{virtual} = [coordinate]

\def\keywordfont{\normalfont\fontsize{9}{11}\selectfont\leftskip2pc\rightskip2pc plus1fill}%
\def\titlefont{\fontsize{18}{22}\selectfont\bfseries\centering}%

\newenvironment{keywords}{\global
  \keywordfont
  \noindent{{\itshape{Keywords\/}}:}
}

\title{\titlefont LQG for portfolio optimization}
\date{}
\author[1]{M. ABEILLE}
\author[2]{E. SERIE}
\author[3]{A. LAZARIC}
\author[4]{X. BROKMANN}
\affil[1]{Inria Lille - Nord Europe, 40 avenue Halley, 59650 Villeneuve d'Ascq, France}
\affil[2]{Capital Fund Management, 23 rue de l'Universite, 75007 Paris, France}
\affil[3]{Inria Lille - Nord Europe, 40 avenue Halley, 59650 Villeneuve d'Ascq, France}
\affil[4]{Variance Capital LLP - 52 Upper Street, London N1 0QH}

\begin{document}

\maketitle

\begin{abstract}

We introduce a generic solver for dynamic portfolio allocation problems when the market exhibits return predictability, price impact and partial observability. We assume that the price modeling can be encoded into a linear state-space and we demonstrate how the problem then falls into the LQG framework. We derive the optimal control policy and introduce analytical tools that preserve the intelligibility of the solution. Furthermore, we link the existence and uniqueness of the optimal controller to a dynamical non-arbitrage criterion. Finally, we illustrate our method using a synthetic portfolio allocation problem.

\end{abstract}


\section*{Introduction}

Modern finance theory is often thought to have started with the mean-variance approach of Markowitz~\citep{markowitz1952portfolio}. This approach provides 
portfolio managers with a systematic treatment of the risk-return tradeoff by maximizing their own utility. 
This started intensive research further to develop the basic mean-variance theory. In particular, it raised questions about the relationship between risk and return, leading to the famous CAPM model~\citep{sharpe1964capital, jensen1972capital} as well as finer modeling for the risk structure and the return predictability~\citep{fama1993common}.
However, one of the limitations of these approaches is their inability to take transaction costs into account: 
in a multi-step setting, performing such a strategy may be highly suboptimal as the rebalancing cost can be worse than the expected gain. 
This observation, of crucial importance for 
practitioners, led to dynamic allocation rules where portfolio managers anticipate this additional cost and track the Markowitz position by constraining the
 turnover (see for instance~\cite{constantinides1979multiperiod, taksar1988diffusion, morton1995optimal, grinold2010signal}).\\
When the volume of the transaction is large compared to the available liquidity, another effect known as price impact induces transaction costs: 
the execution of a large order drastically changes the supply and demand and thus affects the price in an adverse manner. 
The understanding of the market impact and the way to minimize it is an important topic for large investors and a large amount of literature adresses this question from different 
perspectives. 
Motivated by stylized facts and empirical studies which stress that markets digest very slowly modifications induced by large trade ~\citep{bouchaud2008markets, brokmann2014slow}, ~\citep{mastromatteo2014agent, donier2014fully} derive a microstructure based model for the price impact from the dynamic of the latent order book. 
Following the work of~\citep{kyle1985continuous}, another stream of literature considers agent-based model to understand how information is incorporated into the prices and  how it affects the liquidity. 
Finally,~\citep{huberman2004price, gatheral2010no} study the effect of the price impact on the absence of price manipulation and derive various inequalities 
about the shape of the price impact function. On the other hand, a large part of the literature is dedicated to the minimization of the price impact. 
Two types of problem are usually considered: optimal execution (see~\cite{bertsimas1998optimal, almgren2001optimal, gueant2012optimal, obizhaeva2013optimal}) where investors seek to liquidate a given position within a certain period, and optimal allocation (see~\cite{garleanu2009portfolio, de2012optimal, garleanu2013dynamic, kallsen2013general, moreau2014trading}) where investors try to dynamically 
control a portfolio to maximize their risk-profit utility under price impact. \\

In this paper, we focus on the latter and address the optimal portfolio allocation when the market exhibits dynamical return predictability and price impact. We consider an investor who wants to dynamically allocate a portfolio of $N$ assets in order to optimize a muti-horizon Markowitz cost function. While closed-form solutions had been derived for specific returns modeling and quadratic transaction costs (see~\cite{grinold2010signal, garleanu2013dynamic}), we provide an approach for any linear and Markovian returns modeling. Our contributions are twofold. \textbf{1)} We show that the dynamical allocation problem can be turned into a Linear Quadratic Gaussian (LQG) control problem, and thus can be solved efficiently (see~\cite{kalman1960new,kalman1961new}). First, we introduce execution prices into the Profit and Loss (PnL) measure, which allows us to take into account the transaction costs induced by the dynamical price impact effects. As opposed to previous works, it raises naturally quadratic transaction costs and thus, we do not need to add a penalty term to the Markowitz cost function. Secondly, we make use of the linear and Markovian dynamic of the returns, translated in terms of linear state-space structure, to derive the LQG solution. \textbf{2)} We leverage the LQG theory and, in particular, the underlying Riccati equation theory (see~\cite{lancaster1995algebraic}), to rephrase the existence and uniqueness of the solution in terms of stationarity and non-arbitrage properties of the returns dynamic. The aim is to stress the one-to-one relationship between the portfolio and the LQG problem to take advantage of the LQG solvers while maintaining the financial intuition.\\

The structure of the paper is the following. In section~\ref{sec:setting-stage}, we introduce our multi-horizon mean-variance optimization problem and specify the PnL measure that we consider, in terms of decision and execution prices. We briefly present in section~\ref{sec:lq-theory} the LQG theory, both to introduce the notations and to recall the main results that we use. Section~\ref{sec:from-portf-contr} is dedicated to the reformulation of the portfolio allocation problem into an LQG control problem. We also map the LQG existence and uniqueness results  to the stationarity and non-arbitrage properties of the returns dynamic. Finally, in section~\ref{sec:example}, we illustrate our approach on a synthetic example. We specify the structure of the returns modeling, present the associated LQG solution and introduce graphical tools to analyze the optimal strategy.

\section{Setting the stage}
\label{sec:setting-stage}

We introduce in this section the dynamical Markowitz cost function that we design to take into account both the existence of transaction costs induced by price impact effects and the dynamical nature of the return predictability and the price impact. We leverage the standard risk-return trade-off of Markowitz~\citep{markowitz1952portfolio} and extend it in two directions: first, we modify the \textit{Profit and Loss} (PnL) measure adding execution prices which quantify the actual prices at which transactions are made and second, we address the optimization problem \textit{on a trajectory} considering a multi-horizon mean-variance cost function.\\

We consider an investor whose objective is to dynamically construct a portfolio of $N$ assets: at each time step, he can decide to rebalance his portfolio, using his current knowledge, in order to optimize his gain - encoded in the PnL measure - with respect to a cost function that represents the risk-return trade-off. The seminal work of Markowitz suggests to balance between minimizing the PnL variance (e.g., the risk) and maximizing the PnL expectation  (e.g., the return).\\ To take into account the transaction costs, we explicitly do the distinction between the \textit{decision prices} that are observed at every time step and used to take the trading decision and the \textit{execution prices}, at which transactions occur, that are observed after the trade.

\begin{definition}
Let $Q_t$ and $p_t$ be the $N-$dimensional vectors of inventory positions and decision prices at time $t$. Let $\bar{p}_{t+1}$ be the average execution prices of 
trades on the period $[t,t+1[$, we define the Profit and Loss (PnL) between time $t$ and $t+1$ in an accounting way as
\begin{equation}
PnL_{t,t+1} := Q_{t+1}^\transp p_{t+1} -  Q_{t}^\transp p_{t}  -  (Q_{t+1} - Q_{t})^\transp \bar{p}_{t+1}.
\label{eq:pnl_definition}
\end{equation}
\end{definition}
The introduction of two different prices has several major implications: first, the decision prices allow a local valuation of the gain as the quantity $Q_{t}^\transp p_{t}$ represents the current value of the portfolio at time $t$. Since there is no guarantee that the investor can liquidate his positions instantaneously at prices $p_t$, this valuation is unfortunately artificial. However, the introduction of the execution prices balances this artificial valuation since $(Q_{t+1} - Q_{t})^\transp \bar{p}_{t+1}$ represents the true price of a transaction. Considering round-trip trajectories, i.e sequence of $Q_t$ such that $Q_0 = Q_T = 0$, one gets
\begin{equation}
PnL_{0,T} = \sum_{t=0}^{T-1} PnL_{t,t+1} = - \sum_{t = 0}^{T-1} (Q_{t+1} - Q_{t})^\transp \bar{p}_{t+1},
\label{eq:mtm_round_trip}
\end{equation}
in which only the execution prices appear, stressing the fact that such a definition remains globally exact. Finally, the introduction of the execution prices allows us to take into account the transaction costs induced by the price impact directly into the PnL rather than adding a penalty term to the Markowitz cost function as presented in~\citep{garleanu2013dynamic}. We believe that this formulation offers more flexibility as well as a better understanding since the impact effects are modeled directly into the prices dynamic and thus naturally generates transaction costs.
For sake of convenience, we will directly model the prices evolution through the dynamics of their associated decision and execution returns.
\begin{definition}
Let $r^{dec}_{t+1}$ and $r^{exe}_{t+1}$ be respectively the decision and execution returns:
\begin{equation}
\begin{split}
r_{t+1}^{dec} &:= p_{t+1} - p_{t}, \\
r_{t+1}^{exe} &:= p_{t+1} - \bar{p}_{t+1}.
\end{split}
\label{eq:returns_definition}
\end{equation}
Let $q_{t} = Q_{t+1} - Q_{t}$ be the traded quantity at time $t$ (assuming full execution), the PnL equation~\eqref{eq:pnl_definition} can be expressed as
\begin{equation}
PnL_{t,t+1} = \begin{pmatrix}  Q_{t}^\transp &  q_t^\transp \end{pmatrix} \begin{pmatrix} r_{t+1}^{dec} \\ r_{t+1}^{exe} \end{pmatrix}.
\label{eq:pnl_definition_return}
\end{equation}
\end{definition}

To take into account the dynamical nature of the return predictability and price impact effects, the risk-return trade-off must be considered dynamically. We address time dependency by setting the cost function to be the sum over trajectories of local mean-variance Markowitz cost functions.
Formally, let $\mathcal{F}_t$ be a filtration that represents the accumulated information at time $t$, the aim is to find the best policy - i.e., a mapping from observations to positions - which minimizes the average future expected cost:
\begin{equation}
\min_{Q_0,Q_1,\dots} \lim_{T \rightarrow \infty} \frac{1}{T}
\mathbb{E} \left[ \sum_{t=0}^{T-1} \Big( \lambda  \mathbb{V}(PnL_{t,t+1} \hspace{1mm} | \hspace{1mm} \mathcal{F}_t) 
 - \mathbb{E}(PnL_{t,t+1} \hspace{1mm} | \hspace{1mm} \mathcal{F}_t ) \Big) \hspace{2mm} \Big| 
\hspace{2mm} \mathcal{F}_0 \right]
\label{eq:dynamic_markowitz_criterion}
\end{equation}
where the expectation is  taken regarding the stochastic processes $\process{r^{dec}_t}$ and $\process{r^{exe}_t}$ and conditionally to the accumulated information. The 
sequence $\process{Q_t}$ is seeken $\mathcal{F}_t$-adapted which means that decisions are taken at each time step as a function of past and present
 information.\\
We focus here on the average cost per state problem formulation (i.e. we consider the limit when $T \rightarrow \infty$) which is relevant when there is no natural cost-free termination state (see \cite{bertsekas1995dynamic}) and when discounting is inappropriate. Notice that it implies that we are interested in a stationary solution (i.e., time independent) which represents the investor's behavior far from initial conditions. The local mean-variance cost functions are parametrized with a risk-tuning factor $\lambda$ which controls the risk aversion and is chosen by the investor. Since we look for a stationary solution, $\lambda$ is assumed time-independent.

In order to solve this optimization problem, we need to model the environment dynamics, that is how past inputs (trades and noise) influence the future outputs (returns/prices and positions). We are interested here in the case where outputs depend linearly on past inputs and outputs, which can be encoded into a linear state-space formulation \footnote{Any Markovian linear system can be represented as a linear state-space which makes it quite flexible. In particular, it includes standard econometric modeling such as ARMA, ARMAX, factor model etc...}. This is motivated by the fact that it then falls into the LQG framework which provides us with exact solutions together with efficient solvers.\\
In the next section, we recall the LQG theory and then show how - under the linear state-space modeling - one can directly express \eqref{eq:dynamic_markowitz_criterion} into an LQG problem.


\section{LQ theory}
\label{sec:lq-theory}

We summarize in this section the main results of the Linear Quadratic Gaussian (LQG) control problem, whose theory is at the core of our portfolio allocation problem. The interested reader can find a complete analysis of LQG problems in~\citep{bertsekas1995dynamic} from the point of view of dynamic programming, the analysis of the underlying Riccati equations in~\citep{lancaster1995algebraic} and a treatment of the state space approach and the Kalman filtering estimation in~\citep{durbin2012time}, while the knowledgeable reader is only concerned with the notation and can move to Section~\ref{sec:from-portf-contr}.\\

An LQG problem consists in controlling an uncertain stochastic linear system to minimize a quadratic cost function. Formally, given a tuple of matrices $(A,B,C,\Sigma_x, \Sigma_y,Q,R,N)$ and a filtration collecting the observations $\process{y_t}$ of the system $\mathcal{F}_{t} =  \sigma( y_1,\dots,y_t)$, one seeks a stationary policy $\pi$ mapping the current information $\mathcal{F}_t$ to the control $ q_t$ which minimizes the cost function
\begin{equation}
J(\pi) = \lim_{T \rightarrow \infty} \frac{1}{T} \mathbb{E} \left( \sum_{t=1}^T \begin{pmatrix} x_t^\transp & q_t^\transp \end{pmatrix} \begin{pmatrix} Q & N \\ N^\transp & R \end{pmatrix} \begin{pmatrix} x_t \\ q_t \end{pmatrix}  \hspace{2mm} | \hspace{2mm} \mathcal{F}_0, q_s = \pi (\mathcal{F}_s), \forall 1 \leq s \leq T \right),
\label{eq:lqg_cost_function}
\end{equation}
where the dynamical input-output modeling is encoded into a linear state space
\begin{equation}
 \left\{
\begin{array}{llll}
  x_{t+1} &=& A x_t + B q_t + \epsilon_{t+1}^x, & \epsilon_{t+1}^x \sim \mathcal{N}(0,\Sigma_x) \\
  y_{t} &=&  C x_t + \epsilon_{t}^y, & \epsilon_{t}^y \sim \mathcal{N}(0,\Sigma_y) 
\end{array}\right.
\label{eq:lqg_dynamic}
\end{equation}
which models the relationship between inputs $q$, $\epsilon^x$, $\epsilon^y$ and observations (outputs) $y$. The dynamical effects are encoded through the internal state $x_t$ which is not observed in general \footnote{The internal state encodes the memory of the system and can be seen as a latent variable with no physical meaning. See~\cite{durbin2012time} for a discussion about the state-space representation.}. Before presenting the solution, we introduce the notions of \textit{stabilizability} and \textit{detectability} which intuitively ensure the system to be non-exploding and the observation process to be rich enough to estimate the internal state.
\begin{definition}
\label{def:stabilizability_detectability}
Let $A,B,C$ be matrices of size $n\times n$, $n \times p$, $d \times n$. 
\begin{description}[align=right,labelwidth=3cm]
\item [stabilizability] The pair $(A,B)$ is said to be stabilizable if there exists a matrix $K$ of size $p \times n$ such that $A +BK$ is asymptotically stable i.e. if all its eigenvalues lie in the open unit disc,
\item [detectability] the pair $(C,A)$ is said to be detectable if $(A^\transp, C^\transp)$ is stabilizable.
\end{description}
\end{definition}
 
\subsection{Separation Principle}

The main advantage of the LQG is the ability to treat separately the control and the estimation sub-problems thanks to the \textit{Separation Principle}. In the LQG case, under suitable assumptions, this heuristic is exact (see~\cite{bertsekas1995dynamic} section 5.2). The estimation part is first performed using Kalman filtering, and the optimal control for the fully observable case - known as the Linear Quadratic Regulator (LQR) - is then applied to the current estimate.
Formally, let $\hat{x}_{t} = \mathbb{E}(x_t |\mathcal{F}_{t} ) $ be the Kalman filter estimate at time $t$, and $K$ be the optimal control of the LQR problem, the optimal control of the LQG problem is given by 
\begin{equation}
\forall t \geq 1, \hspace{3mm} q_t = \pi ( \mathcal{F}_t) = K \hat{x}_{t}.
\label{eq:separation_principle}
\end{equation}
We illustrate the structure of the LQG controller with the block-diagram in Figure~\ref{fig:lqg_block_diagram}, present the LQR and Kalman results as well as the obtained Closed Loop System (CLS).
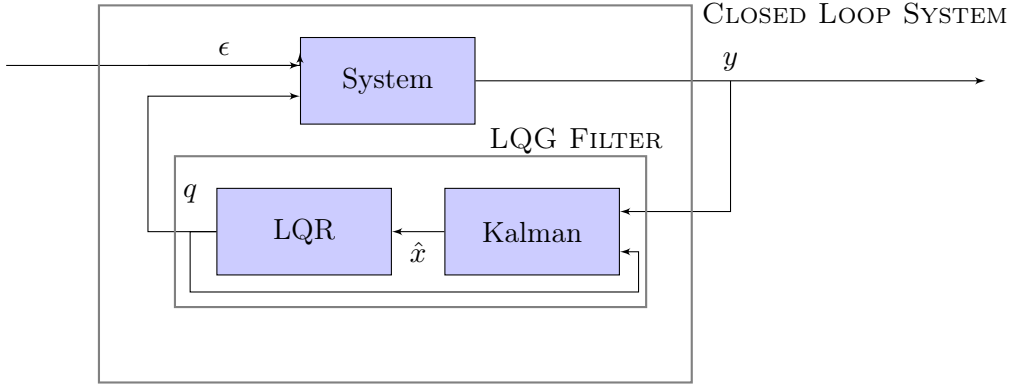
\begin{figure}[!h]
\begin{center}
\begin{tikzpicture}[auto, node distance=2cm,>=latex']

\node[block] (system) {System};
\node[virtual, left=of system.170,font=\bfseries] (aux1) {}; 
\node[left of=aux1] (aux3) {};
\node[virtual, left =of system.190,font=\bfseries] (aux2) {};
\node[name=yy, right of=system, node distance=4cm]{};
\node[right of=yy,font=\bfseries, node distance=4cm] (output) {};
\draw[->] (aux1) -- node [name=eps] {$\epsilon$}(system.170);
\draw[->] (system) -- node [name=y] {$y$}(output);
\node [name=kk,below of=yy]{};
\node[block, name=kalman, left of=kk,node distance=2.1cm]{Kalman};
\node[block, name=controller, left of=kalman, node distance=3cm]{LQR};
\draw [->] (y) |- (kalman.015);
\draw [->] (kalman) -- node [name=xhat] {$\hat{x}$} (controller);
\draw [->] (controller) -| (aux2) |- (system.190);
\draw [->] (aux3) -| (system.170);
\node[virtual, left of=controller, node distance=1.5cm] (aux4) {};
\node[name=q,above of=aux4, node distance = 0.5cm] (q) {$q$};
\node[virtual, below of=aux4,node distance=0.8cm](aux5) {};
\node[virtual, right of=aux5,node distance=5.9cm] (aux6) {};
\draw [->] (controller) -- (aux4) -- (aux5) -- (aux6) |-  (kalman.-015);

\draw [color=gray,thick](-2.8,-3) rectangle (3.4,-1);
\node at (1.2,-0.8) [below=5mm, right=0mm] {\textsc{LQG Filter}};

\draw [color=gray,thick](-3.8,-4) rectangle (4,1);
\node at (4,0.9) [below=5mm, right=0mm] {\textsc{Closed Loop System}};
\end{tikzpicture}
\caption{\label{fig:lqg_block_diagram} Illustration of the LQG controller: the state-space \eqref{eq:lqg_dynamic} defines the block \textit{System} which produces observations $y$ given input noise and control $q$. The observations $y$ are then used to feed the \textit{LQG} filter which consists in two filters - estimation and control - in series together with a feedback on input $q$. Finally, the obtained control $q$ is fed back into the \textit{System}. Thanks to the separation principle, the \textit{Kalman} and \textit{LQR} filter can be solved separately.}
\vspace{-10mm}
\end{center}
\end{figure}

\subsection{Kalman Filter and LQR}

The Kalman filter dates back to the seminal work of Kalman (see~\cite{kalman1960new,kalman1961new}) and consists in estimating the 
distribution of the internal state $x_t$ given measurements $y_t$. The gaussianity of the noise in~\eqref{eq:lqg_dynamic} reduces the estimation to $\hat{x}_{t}  = \mathbb{E}(x_t | \mathcal{F}_t)$ and $\Omega_x = \mathbb{V}( x_t | \mathcal{F}_t).$\footnote{ Since we focus on an average per cost problem \eqref{eq:lqg_cost_function}, we look for a steady-state estimation i.e. far from initial conditions. As a consequence, the variance estimate is time independent.}

\begin{theorem}(see~\cite{lancaster1995algebraic} th.17.5.3)
\label{th:kalman_filter}
Assume that the noise sequences $\process{\epsilon_{t+1}^x}$ and $\process{\epsilon_{t}^y}$ are conditionally Gaussian and mutually independent. Assume that the pair $(A,\Sigma_x)$ is stabilizable and that the pair $(C,A)$ is detectable, then, the steady-state solution of the Kalman filter is given by:

\begin{equation}
\begin{aligned}
&\left\{
\begin{aligned}
& \tilde{x}_{t+1} = A (I - LC) \tilde{x}_{t} + B q_t + A L y_t \\
&\hat{x}_{t} = (I - LC) \tilde{x}_{t} + L y_t
\end{aligned}
\right. \\
 \text{where } \hspace{5mm}
& \tilde{\Omega}_x = \Sigma^x + A \tilde{\Omega}_x A^\transp - A \tilde{\Omega}_x C^\transp \left( \Sigma^y + C \tilde{\Omega}_x C^\transp \right)^{-1}  C \tilde{\Omega}_x A^\transp, \\
 & \Omega_x = \tilde{\Omega}_x - L C \tilde{\Omega}_x, \hspace{5mm}   L = \tilde{\Omega}_x C^\transp \left( \Sigma^y + C \tilde{\Omega}_x C^\transp \right)^{-1}.
\end{aligned}
\label{eq:kalman_state_space}
\end{equation}
The matrix $L$ is called the Kalman gain, and the matrix $\tilde{\Omega}_x = \mathbb{V} (x_t | \mathcal{F}_{t-1})$ is the solution of a Riccati equation. Stabilizability and detectability ensure $\tilde{\Omega}_x$ to be unique, symmetric and semi-definite positive.
\end{theorem}

When the internal state $x_t$ is fully observable i.e. $\mathcal{F}_t-$measurable, the LQG problem collapses into the LQR problem. Namely, one seeks a stationary policy $\pi^{lqr}$ mapping states $\process{x_{t}}$ to controls $\process{q_t}$ in order to minimize the following subproblem:

\begin{equation}
\begin{split}
J^{lqr}(\pi^{lqr}) &= \lim_{T \rightarrow \infty} \frac{1}{T} \mathbb{E} \left( \sum_{t=1}^T \begin{pmatrix} x_t^\transp & q_t^\transp \end{pmatrix} \begin{pmatrix} Q & N \\ N^\transp & R \end{pmatrix} \begin{pmatrix} x_t \\ q_t \end{pmatrix}
\hspace{2mm} | \hspace{2mm} x_1, q_t = \pi^{lqr}(\mathcal{F}_t) \right) \\ 
\text{subject to } & x_{t+1} = A x_{t} + B q_t + \epsilon^x_{t+1}.
\end{split}
\label{eq:lqr_problem}
\end{equation}

\begin{theorem}(see~\cite{lancaster1995algebraic} th.16.6.4)
\label{th:lqr}
Let  $\process{\epsilon^x_t}$ be a $\mathcal{F}_t-$martingale difference sequence. Assume that $(A,B)$ is a stabilizable pair. Assume that the cost matrix $\begin{pmatrix} Q & N \\ N^\transp & R \end{pmatrix}$ is symmetric positive definite\footnote{The existence theory of Riccati solutions has been intensively studied. In particular, two equivalent criterion - one based on matrix pencil regularity, one based on the Popov criterion - can be used to relax the positiveness hypothesis. We discuss this relaxation in the appendix.} , the optimal solution of \eqref{eq:lqr_problem} is given by
\begin{equation}
\begin{split}
q_t &= K x_t, \\
K &= -(R + B^\transp P B)^{-1} (B^\transp P A + N^\transp), \\
P &= Q + A^\transp P A - (A^\transp P B + N ) (R + B^\transp P B)^{-1} ( B^\transp P A  + N^\transp), \\
\end{split}
\label{eq:lqr_solution}
\end{equation}
and  $A + B K$ is asymptotically stable.
\end{theorem}

Symmetrically to the Kalman solution, the matrix $K$ is called the control gain and depends on a matrix $P$ which is the solution of another Riccati equation.

\subsection{Closed Loop System}

Once the system is controlled using a stable filter which maps $y_t$ to $q_t$, we obtain a Closed Loop System (CLS) where input $q_t$ is eliminated. Since every block - or filter - involved in the LQG modeling is encoded using linear state-spaces, we can express the obtained CLS into an augmented state-space (we postpone the derivation in appendix). Formally, let $\mathcal{X}_t = \begin{pmatrix} x_t \\ \hat{x}_{t} \end{pmatrix}$ be the augmented state, the dynamics of the CLS is encoded into \\
\begin{equation}
\left\{
\begin{aligned}
\mathcal{X}_{t+1} &= \mathcal{A} \mathcal{X}_{t} + \mathcal{B} \mathcal{E}_{t+1}, \\
y_t &= \begin{pmatrix} C & 0 \end{pmatrix} \mathcal{X}_{t} + \epsilon_{t}^y,
\end{aligned}
\right.
\label{eq:closed_loop_system}
\end{equation}
with appropriate noise and matrices $\mathcal{E}$, $\mathcal{A}$, $\mathcal{B}$. This formulation offers the advantage of summarizing the behavior of the CLS in a single object. For instance, one can apply standard control tools to~\eqref{eq:closed_loop_system} to obtain analytical expressions for the steady-state variances $V := \mathbb{V}(x_t)$ and $\hat{V} = \mathbb{V}(\hat{x}_{t})$.

\begin{lemma}(see~\cite{lancaster1995algebraic} th.5.3.5)
\label{le:lyapunov}
Let $X_t$ be a stable process with dynamic $X_{t+1} = A X_{t} + \eta_t$ where $\eta_t$ is a martingale difference sequence with respect to $\process{X_t}$ with conditional variance $E$. Then its steady-state variance is given by the Lyapunov equation
\begin{equation}
\mathbb{V}(X_t) = A \mathbb{V}(X_t) A^\transp  + E.
\label{eq:lyapunov}
\end{equation}
Moreover, since $A$ is stable and $E$ semi-definite positive, equation \eqref{eq:lyapunov} admits a unique solution which shares the same signature as $E$.
\end{lemma}
Applying lemma \ref{le:lyapunov} to the closed loop system~\eqref{eq:closed_loop_system} gives $\mathbb{V}(\mathcal{X}_t)$ from which $V$ and $\hat{V}$ can be extracted.


\section{From portfolio control to LQG}
\label{sec:from-portf-contr}

We present in this section the main contribution of the paper. We show that the dynamical Markowitz portfolio optimization problem~\eqref{eq:dynamic_markowitz_criterion} of Section~\ref{sec:setting-stage} can be rephrased as an LQG control problem as soon as the returns dynamics is encoded into a linear state-space. We leverage this reformulation to translate the existence and uniqueness results of the LQG in terms of stability and non-arbitrage properties of the underlying returns modeling. As a consequence, it concentrates the complexity of the portfolio allocation into the returns modeling rather than into the derivation of the solution.

\subsection{LQG construction}

We consider the dynamical portfolio optimization problem~\eqref{eq:dynamic_markowitz_criterion}. We assume that the trades are fully executed on the market i.e. $Q_{t+1} = Q_{t} + q_{t}$ so one seeks a policy $\pi$ mapping observations $\mathcal{F}_{t}$ to trades $\process{q_t}$ which minimizes
\begin{equation}
\lim_{T \rightarrow \infty} \frac{1}{T}
\mathbb{E} \left( \sum_{t=0}^{T-1} \lambda  \mathbb{V}(PnL_{t,t+1} \hspace{1mm} | \hspace{1mm} \mathcal{F}_t) 
 - \mathbb{E}(PnL_{t,t+1} \hspace{1mm} | \hspace{1mm} \mathcal{F}_t ) \hspace{2mm}| 
\hspace{2mm} \mathcal{F}_0, q_t = \pi( \mathcal{F}_{t}) \right),
\label{eq:dynamic_markowitz_trade}
\end{equation}
where $PnL_{t,t+1} = \begin{pmatrix}  Q_{t}^\transp &  q_t^\transp \end{pmatrix} \begin{pmatrix} r_{t+1}^{dec} \\ r_{t+1}^{exe} \end{pmatrix}$. As stated in the introduction, we restrict the input-output modeling to be encoded into a linear state-space. Formally, recalling \eqref{eq:lqg_dynamic}, we assume that there exists a state-space representation
\begin{equation*}
 \left\{
\begin{array}{llll}
  x_{t+1} &=& A x_t + B q_t + \epsilon_{t+1}^x, & \epsilon_{t+1}^x \sim \mathcal{N}(0,\Sigma_x) \\
  y_{t} &=&  C x_t + \epsilon_{t}^y, & \epsilon_{t}^y \sim \mathcal{N}(0,\Sigma_y) 
\end{array}\right.
\end{equation*}
where $y_t$ is the vector of observations available at time $t$ and contains the outputs $(Q_t, r_{t}^{dec}, r_{t}^{exe})$. The noise sequences $\process{\epsilon^{x}_t}$ and $\process{\epsilon^{y}_{t}}$ are assumed to be Gaussian and mutually independent. We denote as $\Pi^{Q}$, $\Pi^{dec}$ and $\Pi^{exe}$ the linear mappings from observations to outputs respectively: 
\begin{equation*}
Q_t = \Pi^{Q} y_t, \hspace{5mm} r_{t}^{dec} = \Pi^{dec} y_t, \hspace{5mm} r_{t}^{exe} = \Pi^{exe} y_t, \hspace{5mm} \forall t \geq 1
\end{equation*}
which implies that $PnL_{t,t+1} = \begin{pmatrix} Q_t^\transp & q_t^\transp \end{pmatrix} \begin{pmatrix} \Pi^{dec} \\ \Pi^{exe} \end{pmatrix} y_{t+1}$. Therefore, one gets,
\begin{equation}
\begin{array}{lll}
\mathbb{E} \left( PnL_{t,t+1}| \mathcal{F}_t \right) &=&  \begin{pmatrix} Q_t^\transp & q_t^\transp \end{pmatrix} \begin{pmatrix} \Pi^{dec} \\ \Pi^{exe} \end{pmatrix} \mathbb{E}(y_{t+1} | \mathcal{F}_{t}), \\
\mathbb{V} \left( PnL_{t,t+1} | \mathcal{F}_t\right) &=& \begin{pmatrix} Q_t^\transp & q_t^\transp \end{pmatrix} \begin{pmatrix} \Pi^{dec} \\ \Pi^{exe} \end{pmatrix} \mathbb{V}(y_{t+1} | \mathcal{F}_{t}) \begin{pmatrix} \Pi^{dec,\transp} & \Pi^{exe,\transp} \end{pmatrix} \begin{pmatrix} Q_t \\ q_t \end{pmatrix}, \\
\mathbb{E}( y_{t+1} | \mathcal{F}_t ) &=& C \left( A \mathbb{E}(x_t | \mathcal{F}_t) + B q_t \right),\\
\mathbb{V}( y_{t+1} | \mathcal{F}_t ) &=& \Sigma^y + C \left( \Sigma^x +  A \Omega_x A^\transp \right)C^\transp.
\end{array}
\label{eq:mean_variance_pnl}
\end{equation}
As opposed to the standard Markowitz approach where the risk (e.g. the variance of the PnL) comes from the unpredictable moves of the market (e.g. the noise sequences $\process{\epsilon^x_t}$, $\process{\epsilon^y_t}$), the risk is here augmented by an additive term which represents the uncertainty about the internal state of the system and hence takes into account the risk coming from the partial observability. This is clear when we look at the expression of $\mathbb{V}( y_{t+1} | \mathcal{F}_t )$: the variance induced by the noise $\Sigma^y + C \Sigma^x C^\transp$ is augmented by the additive variance $ C A \Omega_x A^\transp C^\transp$ which implies that the investor will be more and more risk-averse as the uncertainty about the system increases.\\
Finally, the local cost function of \eqref{eq:dynamic_markowitz_trade} is quadratic in the internal state $x_t$ and the trade $q_t$. Algebraic manipulation leads to:
\begin{equation*}
\lambda  \mathbb{V}(PnL_{t,t+1} \hspace{1mm} | \hspace{1mm} \mathcal{F}_t) 
 - \mathbb{E}(PnL_{t,t+1} | \mathcal{F}_{t} )  = \mathbb{E} \left( \begin{pmatrix} x_t^\transp & q_t^\transp \end{pmatrix} \begin{pmatrix} Q & N \\ N^\transp & R \end{pmatrix} \begin{pmatrix} x_t \\ q_t \end{pmatrix} |\mathcal{F}_{t} \right)
\end{equation*}
where
\begin{equation}
 \begin{pmatrix} Q & N \\ N^\transp & R \end{pmatrix} := \begin{pmatrix}  (\Pi^{Q} C)^\transp \Pi^{dec} & (CA)^\transp \\ \Pi^{exe} & (CB)^\transp \end{pmatrix} 
 \begin{pmatrix}  \lambda \Sigma & -\frac{1}{2} I \\ - \frac{1}{2} I & 0 \end{pmatrix} 
 \begin{pmatrix} \Pi^{dec,\transp} \Pi^Q C & \Pi^{exe,\transp} \\ CA & CB \end{pmatrix} 
 \label{eq:cost_matrices_expression}
\end{equation}
 with $\Sigma = \Sigma^y + C ( \Sigma^x +  A \Omega_x A^\transp)C^\transp$.\\

Since the local cost function is quadratic in the state and the trade, it turns the portfolio allocation problem~\eqref{eq:dynamic_markowitz_trade} into the LQG control problem ~\eqref{eq:lqg_cost_function}. This is implied by the definition of the dynamical Markowitz cost function~\eqref{eq:dynamic_markowitz_criterion} together with the linear state-space structure  without further assumptions on the returns modeling. Therefore, we benefit from a unique solver for a large class of returns modeling. The cost matrix ~\eqref{eq:cost_matrices_expression} is constructed directly from the linear state-space matrices, stressing the advantage of the introduction of the execution prices into the PnL definition: first, it naturally takes transaction costs into account and thus does not require to add it afterward to the cost function\footnote{For instance, in the case of instantaneous impact, we retrieve the same quadratic transaction cost  $q_t^\transp R q_t$ as in~\citep{garleanu2013dynamic}.}. Secondly, it allows for dynamical price impact effects which can be encoded into the internal state and be propagated into the cost function.
Finally, this approach shares the LQG advantages of an input/output formulation which are quantities available to the user, hiding the dynamic complexity into the state-space internal state. As illustrated in section \ref{sec:example}, this matrix formulation does not imply a lack of understanding as the system can be analyzed using standard control theory tools.

\subsection{Existence and uniqueness: a non-arbitrage criterion}
\label{subsec:existence_uniqueness}
To guarantee the existence and uniqueness of a solution to the portfolio allocation problem \eqref{eq:dynamic_markowitz_trade} we make the following assumptions about the state-space modeling \eqref{eq:lqg_dynamic} and discuss it in term of returns modeling. Then, we relax the positiveness assumption of theorem~\ref{th:lqr} and replace it with a non-arbitrage criterion.
The first assumption concerns the validity of the separation principle - at the core of the LQG solution.
\begin{assumption}
\label{as:separation_principle}
\textbf{Separation Principle: } The sequence of noises $\process{\epsilon^x_{t}}$ and $\process{\epsilon_t^y}$ are martingale difference sequences with respect to $\mathcal{F}_{t}$, conditionnally Gaussian with respective variances $\Sigma^x$ and $\Sigma^y$ and mutually independent.
\end{assumption}
The second assumption concerns the stability and the detectability of the state-space \eqref{eq:lqg_dynamic} which are required to guarantee both Kalman and LQR solutions.
\begin{assumption}
\label{as:stabilizability}
\textbf{Stabilizability and detectability} Let \eqref{eq:lqg_dynamic} be the state-space modeling the dynamical relationship between the trade $q$, the noise $\epsilon^x$ and the output $y$.
\begin{enumerate}
\item The pair $(A,\Sigma^x)$ is stabilizable and the pair $(C,A)$ is detectable.
\item The pair $(A,B)$ is stabilizable.
\end{enumerate}
\end{assumption}
The first point deals with the Kalman estimation and hence does not care about the sequence of trades $q$ which can be set to zero while the second point addresses the control structure and does not care about the noise sequences. Intuitively, the stabilizability of $(A,\Sigma^x)$ ensures the returns to be stationary under the absence of trading whereas the stabilizability of $(A,B)$ states that there exists trading policies which make the prices stationary. The detectability assumption guarantees that we can recover the meaningful part of the internal state from the sequence of observations and thus addresses the relevance of the information available to the investor. Hence, we see that the control assumptions meet the usual modeling assumptions of portfolio allocation problems.\\

On the other hand, the positiveness of the cost matrix in theorem~\ref{th:lqr} does not hold by construction of~\eqref{eq:cost_matrices_expression}. However, it can be relaxed provided that a Popov criterion is satisfied (see~\cite{molinari1975,lancaster1995algebraic}). To get more intuition, we translate this into a non-arbitrage criterion in line with~\citep{gatheral2010no}. First, we define formally the notion of round-trip sequence which is of crucial importance in our framework, since, according to~\eqref{eq:mtm_round_trip}, it only deals with the execution prices and thus overcomes the fallacious valuation induced by the decision prices.  
\begin{definition}
\label{def:admissible_round_trip}
We denote as $\mathcal{RT}$ the admissible round-trip sequence of trade $q$ where 
\begin{equation*}
\mathcal{RT} := \{ q = (q_0,q_1,\dots) \in l^1 \cap l^2 \hspace{3mm}  \text{such that} \hspace{3mm} Q = (Q_0,Q_1,\dots)  \in l^1 \cap l^2 \},
\end{equation*}
which implies that
\begin{equation*}
\mathcal{RT} \subset  \{ q = (q_0,q_1,\dots) \in l^1 \cap l^2 \hspace{3mm}  \text{such that} \hspace{3mm} \sum_{t=0}^\infty q_t = 0  \}.
\end{equation*}
\end{definition}
Such definition extends the usual round-trip that we consider, i.e. trade trajectories such that $Q_T = Q_0$, to the infinite horizon setting, by ensuring a suitable decay for the inventory positions. Since theorem~\ref{th:lqr} is concerned about the control part, the non-arbitrage criterion only focuses on the impact of trades on prices and thus does not consider the unpredictable moves of the prices. We postpone the proof of theorem~\ref{th:molinari_round_trip} in appendix~\ref{app:proof_non_arbitrage}.

\begin{theorem}
\label{th:molinari_round_trip}
 Assume that the pair $(A,B)$ is stabilizable and consider the non-stochastic
 $PnL_{t,t+1} = \begin{pmatrix} Q_t^\transp & q_t^\transp \end{pmatrix} \begin{pmatrix} r^{dec}_{t+1} \\ r^{exe}_{t+1} \end{pmatrix}$ where the positions, trades and returns are related by
\begin{equation*}
\left\{
\begin{array}{lll}
x_{t+1} &=& A x_{t} + B q_t,\\
y_{t} &=& C x_{t},\\
Q_t &=& \Pi^{Q} y_t, \hspace{5mm} r_{t}^{dec} = \Pi^{dec} y_t, \hspace{5mm} r_{t}^{exe} = \Pi^{exe} y_t, \hspace{5mm} \forall t \geq 1,
\end{array}
\right.
\end{equation*}
 then for any risk parameter $\lambda \in (0,\infty)$ there exists a (necessarily) unique symmetric stabilizing solution $P$ satisfying~\eqref{eq:lqr_solution} if and only if $\sum_{t=0}^\infty PnL_{t,t+1} \leq 0$ for any $q \in \mathcal{RT}$.
\end{theorem}
Theorem~\ref{th:molinari_round_trip} maps the existence of a \textit{stabilizing} solution i.e. a solution $P$ to the Riccati equation so that the associated control $K$ stabilizes the system with the non-positiveness of the PnL associated with any round-trip trade trajectory. This mapping holds for the non-stochastic portfolio allocation problem and hence can be understood as a non-arbitrage criterion since the absence of noise induces the absence of prices predictability.\\
We have shown in this section how to turn the dynamical Markowitz allocation problem into an LQG control problem. Moreover, we exhibited the one-to-one relationship between underlying assumptions from control theory and the stationarity and non-arbitrage properties of the returns modeling. Apart from imposing the linear state-space structure, we did not specify the returns modeling thus allowing for many different price impact and return predictability dynamics. We illustrate our approach on a synthetic example in section~\ref{sec:example}.


\section{An example with separated alpha and impact}
\label{sec:example}
We illustrate our approach on a complete example of portfolio allocation problem when the return predictability and the impact effects are separated. We first describe the dynamics of the returns and show how to encode it into a state-space formulation. We discuss the LQG solution and stress how one can take advantage of the LQG theory to derive analytical expressions for the average performance of the strategy. Finally, we calibrate the parameters and introduce graphical tools to analyze the modeling and the optimal strategy.

\subsection{Prices and position modeling}
We consider a \lq separable model\rq  ~of the form:
\begin{equation}
r^{dec}_{t+1} = r^{p}_{t+1} + r^{i}_{t+1} + \sigma \epsilon^{r}_{t+1}, \hspace{5mm} \epsilon^{r}_{t+1} \overset{iid}{\sim} \mathcal{N}(0,1),
\label{eq:example_separable_return1}
\end{equation}
where $r^{dec}_{t+1} = p_{t+1} - p_{t}$ is the return of mark-to-market prices, 
$r^{p}$ and $r^{i}$ are two distinct processes modeling the predictable and the impact part respectively.
We assume that the predictable part is known at time $t$ and that it evolves as a mean-reverting process encoded into:
\begin{equation}
\left\{
\begin{array}{lll}
x^{p}_{t+1} &=& (1-\omega_{p}) x^{p}_{t} + \beta_p \epsilon^{p}_{t+1}, \hspace{5mm} \epsilon^{p}_{t+1} \overset{iid}{\sim} \mathcal{N}(0,1), \\
r^{p}_{t+1} &=& x^{p}_{t}.
\label{eq:example_alpha_model}
\end{array}
\right.
\end{equation}

The impact model is supposed to exhibit the following features: a trade moves the price up proportionally to the quantity but then 
relaxes slowly following an exponential decay. We also consider permanent impact which means that the relaxation does not push the price back to its original value. This is encoded in the impact part of the return $r^{i}$:
\begin{equation}
\left\{
\begin{array}{lll}
x^{i}_{t+1} &=& (1 - \omega_i) x^{i}_{t} +\omega_i \beta_i  q_t, \\
r^{i}_{t+1} &=& x^{i}_{t} + \gamma_{i} q_t.
\end{array}
\right.
\label{eq:example_impact_return}
\end{equation}
The state variable $x^{i}_{t}$ represents the memory of  past trades that are forgotten exponentially. Notice that this is not a quantity directly observed but that has to be reconstructed (from the sequence of past trades). We assume that the execution price follows:
\begin{equation}
\begin{split}
\bar{p}_{t+1} &= \eta p_{t} + (1 - \eta) p_{t+1} +  r^{i}_{t+1}, \\
r^{exe}_{t+1} &= p_{t+1} - \bar{p}_{t+1} = \eta r^{dec}_{t+1} - r^{i}_{t+1},
\end{split}
\label{eq:example_execution_price_model}
\end{equation}
where $r^{exe}_{t+1}$ is the execution return. This models a system where the execution is made on average in between the 
decision prices at $t$ and $t+1$ (parametrized by $\eta$) plus an additional cost which is equal 
to the impact part of the return. 
The inventory position follows the simple dynamic where every trades are fully executed on the market:
\begin{equation}
Q_{t+1} = Q_{t} + q_t.
\label{eq:example_position_dynamic}
\end{equation}

\subsection{State space formulation}

To make use of LQG techniques, we concatenate equations~\eqref{eq:example_separable_return1}~\eqref{eq:example_alpha_model}~\eqref{eq:example_impact_return}~\eqref{eq:example_execution_price_model}~\eqref{eq:example_position_dynamic} into a single state-space. Introducing the \lq return states\rq
~$x^{dec}_{t}$, $x^{exe}_{t}$, the internal state $x_t$ and the observation vector $y_t$ as
\begin{equation}
x^{dec}_{t} =  r^{dec}_{t}, \hspace{5mm} x^{exe}_{t} = r^{exe}_{t}, \hspace{5mm}  x_{t} = \begin{pmatrix} Q_t \\ x^{dec}_{t} \\x^{exe}_{t} \\ x^p_{t} \\ x^i_t \end{pmatrix}, \hspace{5mm} y_t = \begin{pmatrix} Q_t \\ r_{t}^{dec} \\ r_{t}^{exe} \\ x_t^{p} \end{pmatrix},
\label{eq:example_internal_state}
\end{equation}
one gets the state-space formulation \eqref{eq:lqg_dynamic}
\begin{equation*}
\left\{
\begin{array}{lllr}
x_{t+1} &=& A x_t + B q_t + \epsilon_{t+1}^x ,& \epsilon_{t+1}^x \overset{iid}{\sim} \mathcal{N}(0,\Sigma^x), \\
y_{t} &=& C x_t + \epsilon_{t}^y, & \epsilon_{t}^y \overset{iid}{\sim} \mathcal{N}(0,\Sigma^y),
\end{array}
\right.
\end{equation*}
where
\begin{equation}
\begin{split}
&A:= \begin{pmatrix} 1 & 0 & 0 & 0 & 0 \\ 0 & 0 & 0 & 1 & 1 \\ 0 & 0 & 0 & \eta & (\eta - 1) \\ 0 & 0 & 0 & (1 - \omega_p) & 0 \\ 0 & 0 & 0 & 0 & (1 - \omega_i) \end{pmatrix}, \hspace{5mm} B := \begin{pmatrix} 1 \\ \gamma_i \\ \gamma_i ( \eta - 1) \\ 0 \\ \omega_i \beta_i \end{pmatrix}, \hspace{5mm} C := \begin{pmatrix} 1 & 0 & 0 & 0 & 0 \\ 0 & 1 & 0 & 0 & 0 \\ 0 & 0 & 1 & 0 & 0 \\ 0 & 0 & 0 & 1 & 0 \end{pmatrix} \\
&\Sigma^x := \begin{pmatrix} 0 & 0 & 0 & 0 & 0 \\ 0 & \sigma^2 & \eta \sigma^2 & 0 & 0 \\ 0 & \eta \sigma^2 & \eta^2 \sigma^2 & 0 & 0 \\ 0 & 0 & 0 & \beta_p^2 & 0 \\ 0 & 0 & 0 & 0 & 0 \end{pmatrix}, \hspace{5mm} \Sigma^y := 0_{4,4}.
\end{split}
\label{eq:dynamic_matrices}
\end{equation}
This example does not exhibit partial observability in the sense that the noise $\epsilon^y$ is degenerated, thus:
\begin{equation*}
\Omega_x = \mathbb{V} (x_t | \mathcal{F}_{t}) = 0, \hspace{5mm}  \mathbb{V}(y_{t+1} | \mathcal{F}_{t}) = C \Sigma^x C^T,
\end{equation*}
and we recover the usual Markowitz variance term which involves the returns unpredictable noise variance $\sigma^2$. Finally, 
thanks to the $PnL$ definition \eqref{eq:pnl_definition_return}, according to equations \eqref{eq:cost_matrices_expression} we obtain the following LQG cost matrices for the portfolio allocation dynamical problem \eqref{eq:dynamic_markowitz_criterion} with risk tuning parameter $\lambda$:
\begin{equation*}
Q := \begin{pmatrix} \lambda \sigma^2 & 0 & 0 & -1/2 & - 1/2 \\ 0 & & & & \\ 0 & & (0) & & \\ -1/2 & & & & \\ -1/2 & & & & \end{pmatrix}, \hspace{5mm} N :=  \begin{pmatrix} \lambda \eta \sigma^2 - 1/2 \gamma_i \\ 0 \\ 0 \\ 0 \\0 \end{pmatrix},\hspace{5mm} 
R := \lambda \eta^2 \sigma^2 - \gamma_i (\eta - 1).
\end{equation*}
If the aggregation of such a modeling is very convenient because it allows the use of the LQG theory, it may suffer from a lack of clarity and be difficult to analyze and understand. This is more and more obvious when the complexity increases as well as the dimension of the system. To overcome this issue, impulse response raises as a convenient tool since it quantifies how an impulse of the inputs (i.e. $\epsilon^x$ and $q$)  propagates through~\eqref{eq:lqg_dynamic} and modifies the outputs (i.e. $p$ and $\bar{p}$). We present and discuss it on figure \ref{fig:example_impulse_price_open} of section~\ref{ssec:graphical_analysis}.

\subsection{LQG solution}

We discuss first assumptions  ~\ref{as:separation_principle} and ~\ref{as:stabilizability} and present the LQG solution. By construction, $\{ \epsilon^x_{t} \}_{t}$ and $\{ \epsilon^y_{t}\}_{t}$ are martingale
difference sequences, conditionally Gaussian and mutually independent ($\epsilon^y_t$ is degenerated here). Hence, the Separation Principle holds and we can address both Kalman and LQR separately.
Looking at the matrices \eqref{eq:dynamic_matrices}, it is clear that both $(A, \Sigma^x)$ and $(A,B)$ are stabilizable pairs: indeed, the input noise $\epsilon^x$ excites the return and alpha part of the state which are based upon mean-reverting processes so $(A,\Sigma^x)$ is stable. On the other hand, to ensure the stabilizability of the control pair $(A,B)$ one just has to exhibit a control policy $q_t = K x_t$ such that the associated closed loop matrix $A + B K$ is asymptotically stable e.g. has all its eigenvalues in $(-1,1)$. Taking for instance $K = ( -0.5,0,0,0,0)$ proves the desired result \footnote{ Since $Q_t$ is the only unstable part of the system, any controller which stabilizes it makes the system stable.}.\\
Since the cost matrix $Q$ is not positive definite, we also have to check the non-arbitrage criterion of theorem~\ref{th:molinari_round_trip}. While this offers a qualitative tool, we adopt here a more practical approach. Because the existence of a symmetric stabilizing solution of the Riccati equation 
\eqref{eq:lqr_solution} implies its uniqueness, it suffices to solve it numerically and check afterward that the solution stabilizes the system. Moreover, we take advantage of efficient solvers using the Van Dooren method~\citep{van1981generalized} which are consistent with the existence results derived.\\

The Kalman filter equations \eqref{eq:kalman_state_space} applied to \eqref{eq:dynamic_matrices} gives
\begin{equation*}
\Omega_x = 0, \hspace{5mm}  \tilde{\Omega}_x = \Sigma^x + A \tilde{\Omega}_x A^\transp, \hspace{5mm} 
L = \begin{pmatrix} 0_{11} & 0_{13} \\ 0_{3,1} & I_{33} \\ 0_{11} & 0_{1,3} \end{pmatrix}, 
\end{equation*}
which implies that $LC = I_{55} - E_{55} - E_{11}$.\footnote{For sake of simplicity, we denote as $E_{ij}$ the zero matrix with only $1$ on the $ij$ position, and as $I_{nn}$ the identity matrix of size $n\times n$.} The estimated state $\hat{x}_t$ follows the dynamic:
\begin{equation}
\left\{
\begin{array}{lll}
\tilde{x}_{t+1} &=& A(E_{11} + E_{55}) \tilde{x}_{t} + B q_t + A L y_t, \\
\hat{x}_{t} &=& (E_{11} + E_{55} ) \tilde{x}_t + L y_t, \\
\end{array}
\right.
\label{eq:kalman_filter}
\end{equation}
where
\begin{equation*}
\hat{x}_{t} = [\tilde{Q}_t, r_{t}^{dec} , r_t^{exe}, x_t^p, \tilde{x}_t^i ]^T \hspace{2mm} \text{ and } \hspace{2mm} \begin{pmatrix} \tilde{Q}_{t+1} \\ \tilde{x}_{t+1}^i  \end{pmatrix} = \begin{pmatrix} 1 & 0 \\ 0 &(1 - \omega_i) \end{pmatrix} \begin{pmatrix} \tilde{Q}_t \\ \tilde{x}_t^i \end{pmatrix} + \begin{pmatrix} 1 \\ \omega_i \beta_i \end{pmatrix} q_t.
\end{equation*} 
Therefore, one sees that the Kalman filter directly copies the observed states ($r_{t}^{dec} , r_t^{exe}, x_t^p$) and reconstructs the unobservable yet non-stochastic hidden state $x_t^i$ (which requires a reconstruction of the observed state $Q_t$). Once the estimated state has been constructed one applies the optimal LQR control: let $P$ denote the (necessarily unique) stabilizing solution of the Riccati equation \eqref{eq:lqr_solution} and $K$ the associated optimal LQR controller, the optimal control is given by
$q_t = K \hat{x}_{t}$.

\subsection{Closed loop system analysis}

To analyze the behavior of the controlled system, we consider the augmented CLS \eqref{eq:closed_loop_system} with internal state $\mathcal{X}_t = \begin{pmatrix} x_t \\ \hat{x}_t \end{pmatrix}$. Now that the trades are fed back on the state, the impulse response of interest is the effect of $\epsilon^{x}$ on the price $p_{t}$ and the position $Q_t$. We present them in figure \ref{fig:example_closed_loop_response}.
To tune the risk-return trade-off of the strategy, it is of crucial importance for practitioners to quantify the impact of the $\lambda$ risk tuning parameter on the performance. Quantity of interest are naturally the average PnL, the average Risk and the average yearly Sharpe. Formally, one wants to compute for a given $\lambda$ 
\begin{equation}
\overline{PnL} := \mathbb{E}(PnL_{t,t+1}^\lambda), \hspace{3mm}
\overline{Risk} := \sqrt{ \mathbb{E}(\mathbb{V}(PnL_{t,t+1}^\lambda |\mathcal{F}_{t})) },  \hspace{3mm}
\overline{Sh} := \sqrt{250} \frac{\overline{PnL}}{\overline{Risk}}.
\label{eq:example_capacity_quantities}
\end{equation}

A naive approach is to make use of Monte Carlo techniques to approximate \eqref{eq:example_capacity_quantities} by
\begin{equation}
\overline{PnL}  \approx \frac{1}{T} \sum_{t=0}^{T-1} \mathbb{E}(PnL_{t,t+1}^\lambda |\mathcal{F}_{t}), \hspace{5mm}  
\overline{Risk}  \approx \sqrt{ \frac{1}{T} \sum_{t=0}^{T-1} \mathbb{V}(PnL_{t,t+1}^\lambda |\mathcal{F}_{t}) }
\label{eq:example_capacity_quantities_MC}
\end{equation}
for a reasonably large $T$. This procedure is time consuming especially when we consider a huge number of 
$\lambda$ values. Thanks to the LQG theory, it is possible however, to derive analytical formulas. Rewriting the PnL as function of 
the state $x_t$ and the estimated state $\hat{x}_{t}$ and noticing from equations \eqref{eq:mean_variance_pnl} that 
\begin{equation}
\begin{split}
\mathbb{E}(\mathbb{E}(PnL^\lambda_{t,t+1} | \mathcal{F}_{t} )) &=  \Tr \left( M_1 \mathbb{V}(\hat{x}_{t} ) \right),\\ 
\mathbb{E}(\mathbb{V}(PnL^\lambda_{t,t+1} | \mathcal{F}_{t} )) &= \Tr \left( M_2^\transp M_3 M_2 
\begin{pmatrix} C \mathbb{V}(x_t) C^\transp + \Sigma^y & C \mathbb{V}(\hat{x}_t) \\ \mathbb{V}(\hat{x}_t ) C^\transp \Sigma^y & \mathbb{V} (\hat{x}_t) \end{pmatrix} 
\right),
\end{split}
\label{eq:average_cost_closed_loop}
\end{equation}
where
\begin{equation*}
\begin{split}
M_1 &:=  [C^\transp \Pi^{Q,\transp} \Pi^{dec} + K^\transp \Pi^{exe}] C (A + B K), \\
M_2 &= \begin{pmatrix} \Pi^Q & 0 \\ 0 & K \end{pmatrix}, \\
M_3 &:= \begin{pmatrix} \Pi^{dec} \\ \Pi^{exe} \end{pmatrix} [C ( \Sigma^x +  A \Omega A^\transp )C^\transp + \Sigma^y ] \begin{pmatrix} \Pi^{dec,\transp} & \Pi^{exe,\transp} \end{pmatrix},
\end{split}
\end{equation*}

the complexity of computing \eqref{eq:average_cost_closed_loop} lies in the computation of $V := \mathbb{V}(x_t)$ and $\tilde{V} := \mathbb{V}(\hat{x}_t)$. Thanks to lemma \ref{le:lyapunov} this is achieved easily solving a Lyapunov equation. We plot the results on figure~\ref{fig:example_capacity_curves}.

\subsection{Numerical application and graphical analysis}
\label{ssec:graphical_analysis}

We consider the following parameters: we assume that the predictor has a characteristic time scale of $10$ units of time (say days) which corresponds to $\omega_p = 0.1$ and has an associated Markowitz strategy (e.g. without considering impact) with an annualized Sharpe $Sh_y$ of $3$. The Markowitz annualized sharpe is given by\begin{equation*}
Sh_y = \frac{\beta_p (1 - \omega_p) } {\sigma \sqrt{1 - (1- \omega_p)^2}} \sqrt{250}.
\end{equation*}
Therefore, using $\sigma = 2e-2$ and $\rho^{\alpha} = 0.9$, $Sh_y = 3$ corresponds to $\beta_p = 1.8e-3$.\\
We tune the impact parameters assuming a decay of characteristic time of 20 days which corresponds to $\omega_i = 0.2$. The $\gamma_i$ parameter controls the intensity of the price impact which is usually described by $\gamma_{i} = \frac{Y \sigma}{V_t}$ where $\sigma$ is the volatility of the stock, $V_t$ is the 
daily volume traded on the market that represents the liquidity and $Y$ is an adimensionate constant called Y-ratio of order unity.
 For sake of simplicity we set $V_t = 1$ without loss of generality ($q_t$ is then expressed in term of a fraction of the daily volume), and 
$Y = 3$ so that $\gamma_{i} = 0.06$. This means that a trade buying 1\% of the market volume impacts the price up by 6bps. We tune $\beta_{i}$ such that the price has a permanent impact of $20\%$: 
$\beta_{i} = -\gamma_{i} (1 - 0.2) = -0.0048$. Finally, we assume here that $\eta = 0.5$ for sake of simplicity.\\

To analyze the open loop system, we separately plot in figure \ref{fig:example_impulse_price_open} the impulse response of the decision price to the predictor noise (contained in $\epsilon^x$) and to the trade input $q$. As the execution price only matters when a transaction occurs, we cannot analyze its dynamics through impulse responses but focus on round-trip trajectories which make the reference price vanish (see equation~\eqref{eq:mtm_round_trip}). Figure ~\ref{fig:example_round_trip_price_open}  shows three different graphical representations of such a round trip. The last one is a parametric prices versus position graph and  stresses the hysteresis of the system due to the persistence (through decay and permanent impact) of past trades. The area intercepted by the execution price trajectory represents the PnL and is signed according to the circular motion: when it is clockwise, the sign is negative and the trajectory incurs a loss whereas it is a gain for anti-clockwise motion.\\
\begin{figure}[!ht]
\hspace{-1cm}
    \includegraphics[scale=0.4]{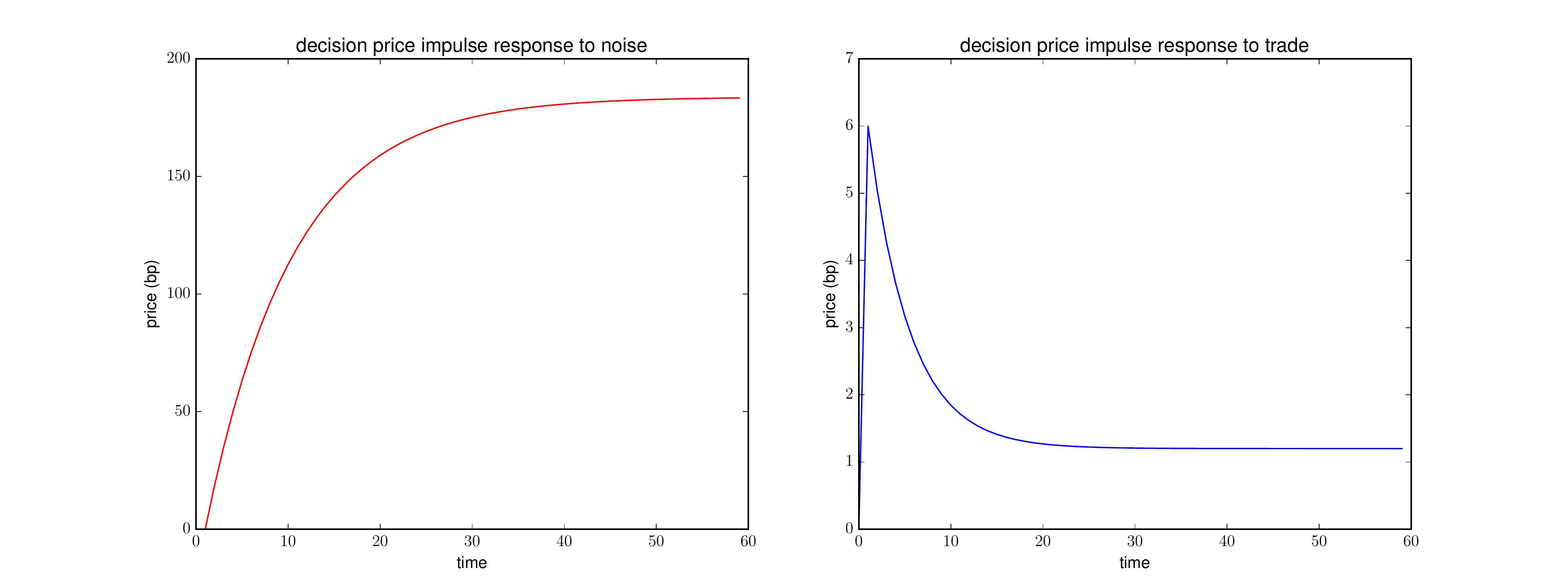}
      \caption{\label{fig:example_impulse_price_open} Impulse response of the decision price. \textit{Left:} An impulse of the predictor noise pushes the price up and the growth is exponential thanks to the auto-regressive modeling of the return. \textit{Right:} The impulse response of a trade summarizes the impact modeling. First, the price is pushed up by the instantaneous impact. Then it decreases exponentially. Finally, its terminal value is higher because of permanent impact.}
\end{figure}
\begin{figure}[!ht]
\hspace{-2.8cm}
  \includegraphics[scale=0.33]{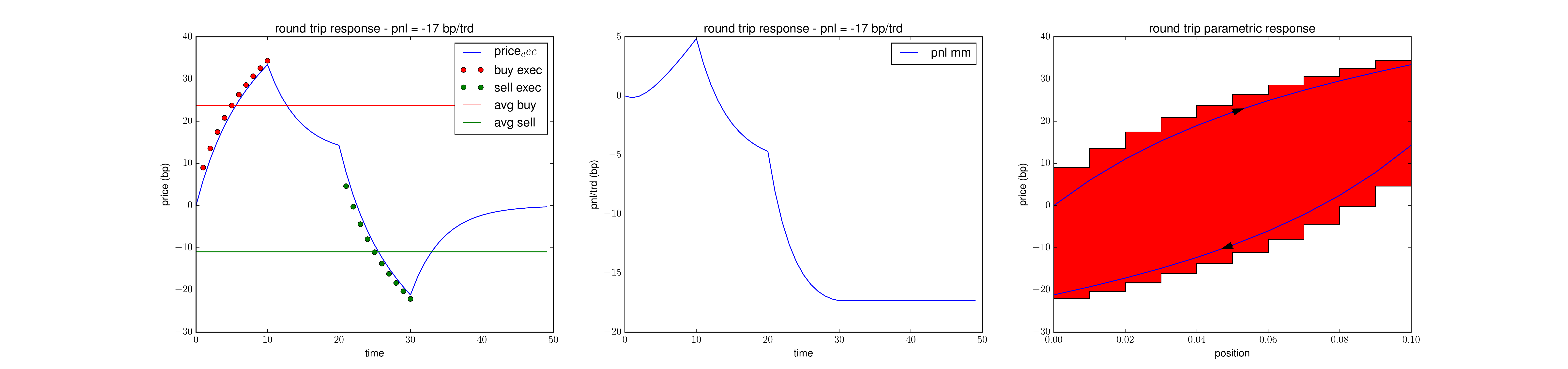}
    \caption{\label{fig:example_round_trip_price_open} Round trip response of the open loop system. We buy 1\% of the daily volume the first 10 days, do nothing the following 10 days and sell back the position the same way. \textit{Left:} Decision and execution prices trajectories. The decision price is pushed up by the purchase and exhibits a concavity induced by the relaxation of past trades. The execution is, on average, in between the two decision prices and augmented by a cost which decreases as the execution takes place. \textit{Center: } PnL trajectory. Since the impact effects act in an adverse manner, it makes the practitioner 'buy high, sell low' and hence, the global PnL is negative. If the PnL seems to rise as the purchase takes place, it is purely artificial and stresses the danger of local valuation. \textit{Right:}  Parametric prices versus position graph. The continuous blue line represents the decision price while the piecewise constant graph represents the execution price. The intercepted red area represents the PnL.}
\end{figure}
\newpage
The dynamics of the CLS, and hence the behavior of the controller, can again be analyzed through impulse responses. The input of interest is the predictor noise and the outputs are the prices and the position. We present it in figure~\ref{fig:example_closed_loop_response}. Moreover, since the stabilizing property of the optimal controller makes every impulse responses a round trip for the CLS, we also plot the hysteresis graph associated to it.\\
\begin{figure}[!ht]
\hspace{-2.8cm}
  \includegraphics[scale=0.33]{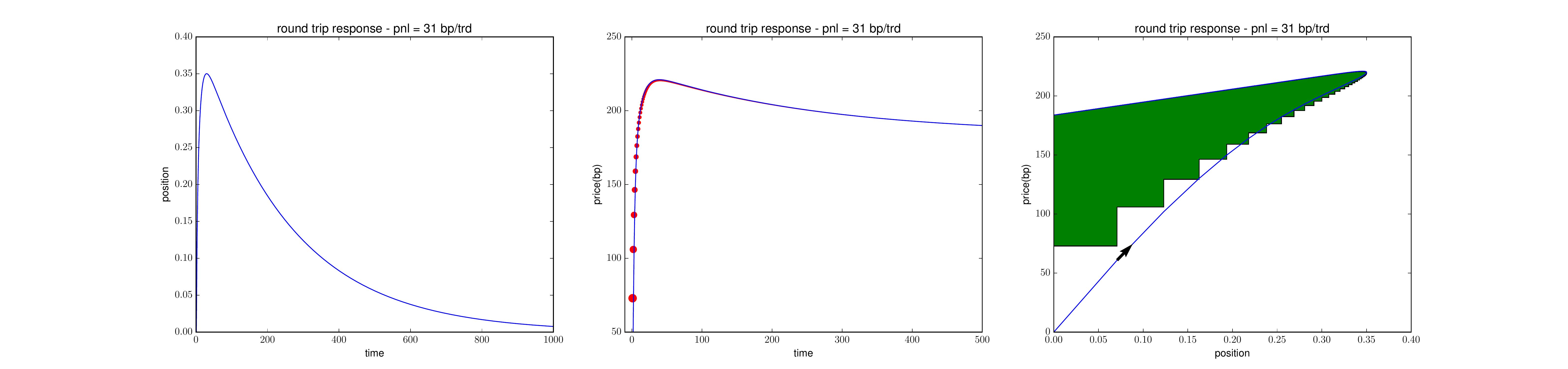}
    \caption{\label{fig:example_closed_loop_response} Impulse response of the closed loop system. \textit{Left:} Position trajectory. In order to capture the upper move of the price (induced by the predictor noise impulse), the position increases fast and decreases to zero very slowly to minimize the cost induced by price impact. The rate of the decay is the optimal trade-off between reducing the risk (selling the position) and keeping the transaction cost low (holding the position). \textit{Center:} 
    Decision and execution price trajectories. The exponential growth of the decision price (induced by the impulse of the predictor noise) is accelerated by the purchase. The permanent impact effect induces overshooting i.e. the price is pushed higher than its natural terminal value. \textit{Right:} Parametric prices versus position graph. The PnL is positive (the circular motion is anti-clockwise) since the round trip is induced by a predictable move of the price.}
    \end{figure}

The  $\lambda$ parameter (chosen by the investor) controls the risk-return trade-off of the optimal strategy. Intuitively, it affects the volume of the inventory position as well as the rate of the decay. This is stressed by figure \ref{fig:example_lambda_impulse_response} where we show the impulse response of the price and position 
for high and low value of $\lambda$. For sake of clarity, we neglect here the permanent impact effect (which induces overshooting). 
To analyze the influence of $\lambda$ on the average performance, thanks to \eqref{eq:average_cost_closed_loop}, we compute and draw PnL vs Risk and Sharpe vs Risk as parametric functions of $\lambda$. We also draw Monte Carlo estimate for trajectories of length $T=5000$.\\

\begin{figure}[!ht]
\hspace{-2.8cm}
  \includegraphics[scale=0.33]{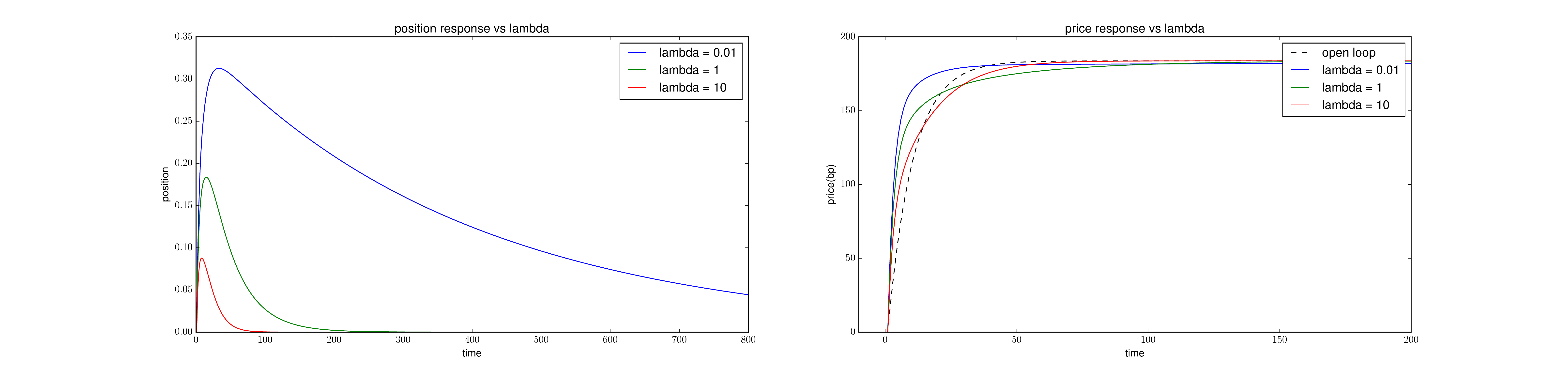}
    \caption{\label{fig:example_lambda_impulse_response} Sensibility to the $\lambda$ parameter. \textit{Left:} Position trajectories. The more aggressive (e.g. the less risk averse) the trader is, the larger inventory positions are and the slower the relaxation is. \textit{Right:} Decision price trajectories. Given 
a natural predictable move of the price (the black dashed line), the trading activity which tries to capture gain from this move 
acts as an arbitrage effect: the more aggressive the trader is, the quicker prices are pushed to their expected value.}
\end{figure}
\begin{figure}[!ht]
\hspace{-2.8cm}
\begin{center}
   \includegraphics[scale=0.33]{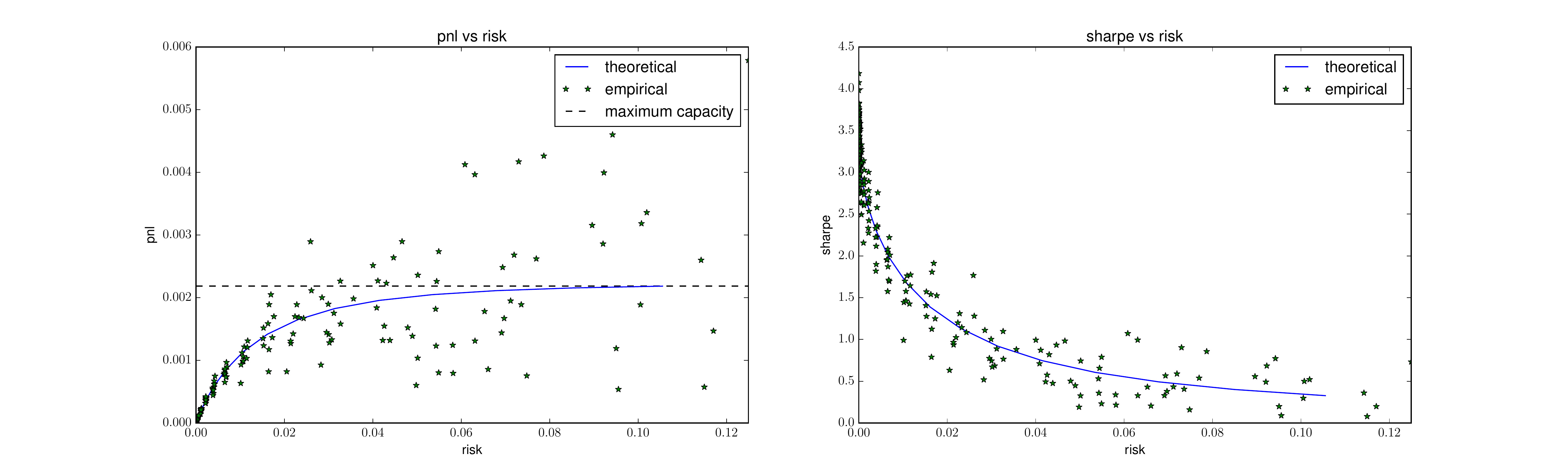}
\end{center}
     \caption{\label{fig:example_capacity_curves} Capacity curves. \textit{Left:} PnL versus Risk. As $\lambda$ decreases, the strategy becomes more and more aggressive. While the PnL increases with the level of risk, it reaches a limit value that we call the maximum capacity of the system: regardless how much risk the trader agrees to take, he cannot increase its gain because the amplitude of the transactions becomes so huge that price impact effects induce prohibitive costs. \textit{Right:} Yearly Sharpe versus Risk. 
The zero risk corresponds to the Markowitz Sharpe of 3 which is unreachable because of price impact: as soon as the trader invests in the strategy (e.g. 
takes more and more risk) the Sharpe is deteriorated by the adverse moves of the price. This explains the decreasing property of the graph.}
\end{figure}

We presented in this toy example how to optimally control a portfolio when the price impact exhibits impact decay and permanent impact. The use 
of the LQG theory allowed us to focus only on the price modeling as the optimal control is provided by the Riccati equation. The obtained 
controller however involves non-trivial quantities that are no longer explicit. To overcome the lack of understanding of such an implicit 
solution, we provided tools to analyze the open loop as well as the closed loop system.

\newpage
\section*{Conclusion}
\label{sec:conclusion}
We have presented here a complete framework to generically solve dynamical Markowitz allocation problem in the presence of return predictability and impact with partial observability - when the prices modeling is encoded into a linear state-space. We derive this formulation based on an accounting formula for the PnL which offers the advantage of concentrating the modeling part into the price dynamics. This approach induces naturally transaction costs since the impact effects are taken into account in the dynamical model. The optimal strategy is then computed using standard solvers from the LQG theory. While explicit formulas are no longer available in general, we introduce tools - both graphical and theoretical - to analyze the solution.\\
Our approach is based upon the LQG and the state-space theory which allows various developments. While we focus in this article on a minimal state-space formulation for the sake of simplicity, several extensions can be solved in the same way. For instance, it is possible to correlate the noise sequences $\epsilon^x$ and $\epsilon^y$, to implement non-exponential decay for the alpha and impact, to merge the alpha and impact modeling instead of considering it separately etc...
Additionally, the partial observability feature allows - among other things - to study the robustness of the allocation policy. Thanks to generic state-space operation (serialization, feedback) it is possible to plug a LQG controller into any system which shares the same input/output structure. Investors can therefore analyze how their controller - based on their belief - would behave in a different market model. Finally, our approach can be easily implemented using existing control library, taking advantage of state-space formulation together with powerful Riccati solvers.\\
We focus here on solving the allocation problem, assuming that the investor has access to the model of prices. However, as opposed to classical physics where dynamics are ruled by laws of nature, financial markets dynamics has to be estimated and calibrated. This open and difficult problem can be studied in this LQG framework and leads to interesting future developments. $H_\infty$ theory and robust control could address this issue in a pessimistic manner, while reinforcement learning technics could be used to learn and control simultaneously. We will discuss the latter in a forthcoming paper.

\section*{Acknowledgements}

The authors would like to warmly thank J-P. Bouchaud, J. Bun, J. Donier, C.A. Lehalle and G. Zerah for insightful discussions and remarks.


\bibliographystyle{plainnat}
\bibliography{biblio}

\appendices
\section{Derivation of the Closed-Loop System}

We detail here the derivation of the CLS which encodes the global dynamics of the system~\eqref{eq:lqg_dynamic} once controlled with an LQG controller. Formally, recalling~\eqref{eq:lqg_dynamic}, one has:
\begin{equation*}
 \left\{
\begin{array}{llll}
  x_{t+1} &=& A x_t + B q_t + \epsilon_{t+1}^x, & \epsilon_{t+1}^x \sim \mathcal{N}(0,\Sigma_x) \\
  y_{t} &=&  C x_t + \epsilon_{t}^y, & \epsilon_{t}^y \sim \mathcal{N}(0,\Sigma_y) 
\end{array}\right.
\end{equation*}
while the LQG controller consists in the Kalman estimation state-space~\eqref{eq:kalman_state_space}:
\begin{equation*}
\left\{
\begin{array}{lll}
\tilde{x}_{t+1} &=& A (I - LC) \tilde{x}_{t} + B q_t + A L y_t, \\
\hat{x}_{t} &=& (I - LC) \tilde{x}_{t} + L y_t,
\end{array}\right.
\end{equation*}
together with the LQR control policy $q_t = K \hat{x}_t$. First, we merge the two state-space to reformulate the dynamics of the estimated state $\hat{x}_t$:
\begin{equation*}
\begin{split}
\hat{x}_{t+1} &= (I - L C) \tilde{x}_{t+1} + L y_{t+1} \\
&= (I - LC) (A \hat{x}_t + B K \hat{x}_t ) + L ( C x_{t+1} + \epsilon^y_{t+1}) \\
& = (I - LC)(A + BK) \hat{x}_{t} + LCAx_{t} + LC B K \hat{x}_t + LC \epsilon^x_{t+1} + L \epsilon^y_{t+1} \\
&= (A + BK) \hat{x}_{t} + LCA(x_t - \hat{x}_{t} ) + LC \epsilon^x_{t+1} + L \epsilon^y_{t+1} \\
\end{split}
\end{equation*}
and we add it to the state-space~\eqref{eq:lqg_dynamic}. Formally, we define the augmented state $\mathcal{X}_{t} = \begin{pmatrix} x_t \\ \hat{x}_t \end{pmatrix}$ whose dynamics is encoded into~\eqref{eq:closed_loop_system}:
\begin{equation*}
\left\{
\begin{aligned}
\mathcal{X}_{t+1} &= \mathcal{A} \mathcal{X}_{t} + \mathcal{B} \mathcal{E}_{t+1}, \\
y_t &= \begin{pmatrix} C & 0 \end{pmatrix} \mathcal{X}_{t} + \epsilon_{t}^y,
\end{aligned}
\right.
\end{equation*}
where 
\begin{equation*}
\mathcal{E}_{t} = \begin{pmatrix} \epsilon^x_{t} \\ \epsilon^y_{t} \end{pmatrix}, \hspace{5mm} 
\mathcal{B} = \begin{pmatrix} I & 0 \\ LC & L \end{pmatrix}, \hspace{5mm}
\mathcal{A} = \begin{pmatrix} A & BK \\ LCA & A+BK - LCA \end{pmatrix}.
\end{equation*}
This derivation is a specific case of a generic property of state-spaces: the structure is stable under \lq standard\rq ~operations such as series, feedback and parallelization. This stresses the advantage of the state-space formulation since we can summarize the dynamics of the CLS in a single object which turns to be just a bigger state-space. Notice that several formulations can be derived, depending on the choice of the internal state, the choice of the output etc... 

\section{Proof of theorem~\ref{th:molinari_round_trip}}
\label{app:proof_non_arbitrage}

We derive here the proof of theorem~\ref{th:molinari_round_trip} which maps the existence and uniqueness guarantee of the LQR solution to a non-arbitrage criterion. The proof is structured as follow: first, we present the Popov criterion (see~\cite{molinari1975}) which guarantee the existence and uniqueness of a solution to the Riccati equation. Second, we show that the deterministic and stochastic LQR share the same Riccati equation and hence, share the same existence and uniqueness condition. Then, we translate the Popov frequency domain criterion in terms of the cost function of the deterministic LQR and thus, in terms of non-arbitrage for admissible trade sequence. Finally, we show that the set of admissible trade sequence is the set of round-trip sequence.\\

\subsection{The Popov criterion}
Since the cost matrix $\begin{pmatrix} Q & N \\ N^\transp & R \end{pmatrix}$ associated with the LQR problem for portfolio construction presented in Section~\ref{sec:from-portf-contr} is not positive definite by construction, the usual guarantee for the Riccati equation solution is violated here. However, the existence of a unique admissible solution of~\eqref{eq:lqr_solution} can still be provided using the Popov criterion. Introducing the hermitian matrix:
\begin{equation}
\Phi(z) = \begin{pmatrix} (Iz^{-1} - A)^{-1} B \\ I \end{pmatrix}^\prime 
\begin{pmatrix} Q & N^\prime \\ N & R \end{pmatrix} \begin{pmatrix} (Iz - A)^{-1} B \\ I \end{pmatrix},
\label{eq:popov_definition}
\end{equation}
and 
\begin{equation}
\begin{split}
\Phi_K(z) &= Y_K^\prime (z^{-1}) \Phi(z) Y_K(z),\\
Y_K(z) &= I + K ( Iz - A - B K)^{-1} B,
\end{split}
\label{eq:popov_stable_definition}
\end{equation}
from~\citep{molinari1975}, we have the following theorem:

\begin{theorem}
\label{th:molinari}
 Assume that the pair $(A,B)$ is stabilizable 
  then there exists a (necessarily) unique symmetric stabilizing solution $P$ satisfying 
\eqref{eq:lqr_solution} if and only if for some (and hence all) $K$ such that $A + B K$ is asymptotically stable, 
$\Phi_K(z) > 0$ for all $|z| = 1$, $z \in \mathbb{C}$.
\end{theorem}

This frequency-domain criterion guarantees the global convexity of the problem based on a fairly complete existence theory 
(see~\cite{molinari1975, ionescu1997general, van1981generalized, wimmer1984algebraic}). The main drawback however is the use of the frequency-domain 
method involved. To get a better intuition about the existence of optimal solution, we link here the Popov criterion to a non-dynamical arbitrage 
criterion in line with~\citep{gatheral2010no}.\\

\subsection{Deterministic LQR}

First, let's notice that the LQR solution of \eqref{eq:lqr_problem} involves the same Riccati equation  (see~\cite{bertsekas1995dynamic}) - and hence shares the same conditions - than the deterministic LQR problem \eqref{eq:lqr_problem_deterministic}:

\begin{equation}
 \begin{aligned}
\underset{\{q_t\}_{t=1,\dots,\infty} \in \mathcal{Q}}{\text{minimize}} \hspace{3mm} \tilde{J}(q_0,q_1,\dots)  &:= \sum_{t=0}^{\infty} 
x_{t}^\prime Q x_{t} + 2 x_{t}^\prime N q_t + q_t^\prime R q_t,  \\
\text{subject to} \hspace{3mm} &
\begin{array}{lll}
  x_{t+1} &=& A x_{t} + B q_{t},\\
\end{array}&
\end{aligned}
\label{eq:lqr_problem_deterministic}
\end{equation}
 where $\mathcal{Q} := \{ q = (q_0,q_1,\dots) \in l^1 \cap l^2 \hspace{3mm}  
\text{such that} \hspace{3mm} x = (x_0,x_1,\dots) \in l^1 \cap l^2 \}$ is the admissible control 
space which are the stabilizing sequences. \\
Indeed, the stochastic LQR problem cost function is defined with an expectation regarding the noise process $\process{\epsilon^x_t}$ which is of zero conditional mean. Thanks to the linear structure of the dynamics, the noises vanish within the Bellman equation which is then the same as the one of the deterministic LQR. As a result, we can apply the Popov criterion on the deterministic problem to ensure the existence and uniqueness of a solution to the stochastic one.

\subsection{From frequency to time domain}

We now state the first corollary which derives directly from theorem \ref{th:molinari}:
\begin{corollary}
\label{co:molinari_stabilizing}
 Assume that the pair $(A,B)$ is stabilizable 
  then there exists a (necessarily) unique symmetric stabilizing solution $P$ satisfying 
\eqref{eq:lqr_solution} if and only if $\tilde{J}(q_0,q_1,\dots) > 0$ for any $q \in \mathcal{Q}$.
\end{corollary}

The proof is straightforward using the z-transform theory. Let $q \in \mathcal{Q}$ be any admissible control sequence and denote as $q(z)$ and $x(z)$ the z-transform of the control sequence and associated state sequence respectively. Then, applying the z-transform theory to \eqref{eq:lqr_problem_deterministic} and using Parceval's theorem leads to:
\begin{equation}
\tilde{J}(q) = \oint_{|z| = 1} q^\transp(z^{-1}) \Phi(z) q(z) dz.
\label{eq:parceval_q}
\end{equation}
The following lemma provides another description for the admissible sequence $q$:

\begin{lemma}
\label{le:admissible_control_sequence}
$q \in \mathcal{Q}$ if and only if there exists a stable control $K$ e.g. such that $A + BK$ is stable, and a sequence $v \in l^1 \cap l^2$ such that 
\begin{equation*}
q_t = K x_t + v_t, \hspace{5mm} \forall t \geq 0.
\end{equation*}
\end{lemma}
\begin{proof}
Let $K$ be a stable control and define $v_t = q_t - K x_t$ for all $t \geq 1$. By definition of $\mathcal{Q}$, $q$ and $x$ belong to $l^1 \cap l^2$ and so does $v$.\\ On the other hand, let $v \in l^1 \cap l^2$, $K$ be a stable control and define $q_t = K x_t + v_t$ for all $t \geq 1$. Then, $x_{t+1} = (A + BK) x_t + v_t$ and since $A + BK$ is stable, $v \in l^1 \cap l^2$ implies that $x \in l^1\cap l^2$ and so does $q$.
\end{proof}

We make use of Lemma~\ref{le:admissible_control_sequence} to rephrase~\eqref{eq:parceval_q} in terms of $v$ sequence: for any $q \in \mathcal{Q}$, let $K$ be a stable control and $v \in l^1\cap l^2$ sequence such that $q_t = K x_t + v_t$. Denoting as $v(z)$ the z-transform of $v$, the z-transform $q(z)$ is:
\begin{equation*}
q(z) = Y_K(z) v(z).
\end{equation*}
Finally, equation~\eqref{eq:parceval_q} becomes:
\begin{equation}
\tilde{J}(q) = \oint_{|z| = 1} q^\transp(z^{-1}) \Phi(z) q(z) dz = \oint_{|z| = 1} v^\transp(z^{-1}) \Phi_K(z) v(z).
\label{eq:parceval_v}
\end{equation}
Therefore, rephrasing theorem~\ref{th:molinari}, there exists a unique symmetric stabilizing solution to the Riccati equation if and only if for some (and hence all) $K$ such that $A+ BK$ is stable $\Phi_K(z) > 0$ for all $|z| =1$ if and only if $\oint_{|z| = 1} v^\transp(z^{-1}) \Phi_K(z) v(z) > 0$ for all $v \in l^1 \cap l^2$ if and only if $\tilde{J}(q) > 0$ for any $q \in \mathcal{Q}$.\\

\subsection{From admissible trade sequence to round-trip}

Finally, to prove theorem~\ref{th:molinari_round_trip}, one just has to show that the set of admissible sequence coincides with the one of round-trip trajectories. Recalling  definition \eqref{def:admissible_round_trip}, one has $$\mathcal{RT} := \{ q = (q_0,q_1,\dots) \in l^1 \cap l^2 \hspace{3mm} \text{such that} \hspace{3mm} Q = (Q_0,Q_1,\dots) \in l^1 \cap l^2\}.$$
By definition, $\mathcal{Q} \subset \mathcal{RT}$ so we just need to prove that for any $q \in \mathcal{RT}$, the associated state sequence is such that $x \in l^1\cap l^2$. To do so, we denote as before $q(z)$, $Q(z)$ and $x(z)$ the z-transform of $\process{q_t}$, $\process{Q_t}$ and $\process{x_t}$ respectively. Then one has:
\begin{equation*}
\begin{split}
x(z) &= (Iz - A)^{-1} B q(z), \\
x(z) &= (Iz - A)^{-1} B (z-1) Q(z).
\end{split}
\end{equation*}
Multiplying by $(z-1)$ and taking the limit when $z \rightarrow 1$ one has, since $(z-1) (Iz-A)^{-1}$ converges to a constant matrix (finite) $H$:
\begin{equation*}
(z-1) (Iz-A)^{-1} \underset{z \rightarrow 1}{\longrightarrow} H < \infty \hspace{5mm} \Longrightarrow \hspace{5mm}
(z-1) x(z)  \underset{z \rightarrow 1}{\sim} H (z-1) Q(z).
\end{equation*}
Thanks to the final value theorem, one gets $x_t \underset{t \rightarrow \infty} {\sim} Q_t$ and since $Q \in l^1 \cap l^2$ by definition of $\mathcal{RT}$ so does $x$. As a consequence, $\mathcal{Q} = \mathcal{RT}$.

\subsection{Plugging everything together}

Thanks to the previous steps, we have that there exists a unique solution to the Riccati equation (provided that $(A,B)$ is stabilizable) if and only if, for any $q \in \mathcal{RT}$, $\tilde{J}(q) > 0$. Noticing that in the absence of noise, the cost function is equal, in term of portfolio allocation to
\begin{equation*}
\tilde{J}(q_0,q_1,\dots)  = \sum_{t=0}^{\infty} - PnL_{t,t+1} 
\end{equation*}
proves theorem~\ref{th:molinari_round_trip}. Because in the absence of noise there is no price predictability, this criterion states that every round-trip must be non-profitable so that impact effects are modeled to act in an adverse manner. Under such guarantee, a unique solution to the portfolio allocation exists.

\end{document}